\documentclass[copyright]{eptcs}
\providecommand{\event}{17th International Conference on Automata and Formal Languages} 
\usepackage{mathrsfs}

\usepackage{iftex}

\ifpdf
  \usepackage{underscore}         
  \usepackage[T1]{fontenc}     
\else
  \usepackage{breakurl}
\fi

\usepackage{verbatim}
\usepackage{amsmath,amsthm,amssymb,bbm,enumerate,mathrsfs,mathtools}
\usepackage{graphicx}
\usepackage{xcolor}

\title{Algebraic Characterizations for Minors of Finite Graphs via Flow Transformation Monoid Division and Embedding} 
\author{Amena Assem$^{1}$ \qquad Hanna Derets$^{2}$ \qquad Chrystopher L. Nehaniv$^{1,3,2}$
\institute{$^1$Department of Systems Design Engineering\\ $^2$David R. Cheriton School of Computer Science\\
$^3$Department of Electrical \& Computer Engineering\\
University of Waterloo, Ontario N2L 3G1, Canada}
\email{\{a36mahmo, h2derets, chrystopher.nehaniv\}@uwaterloo.ca}
}

\def\titlerunning{Algebraic Characterizations for Minors of Finite Graphs via Flow Transformation Monoids}
\def\authorrunning{A. Assem, H. Derets \& C. L. Nehaniv}

\theoremstyle{plain}
\newtheorem{theorem}{Theorem}[section]

\newtheorem{corollary}[theorem]{Corollary}
\newtheorem{lemma}[theorem]{Lemma}

\newcommand{\flowTM}{{\mathbb S}}
\theoremstyle{definition}
\newtheorem{remark}[theorem]{Remark}
\newtheorem{definition}[theorem]{Definition}

\newtheorem{notation}[theorem]{Notation}

\newcommand{\lifts}{\mathrm{lifts}}
\renewcommand{\Im}{\mathrm{Im}}

\begin{document}
\maketitle

\begin{abstract}
We prove three theorems on the flow monoids of finite graphs. First, we show that a non-empty finite graph $\Gamma=(V,E)$ is connected if and only if its flow monoid contains a constant map on $V$, equivalently, if and only if it contains all constant maps on $V$. Second, we give a new characterization of graph minors in terms of division of flow transformation monoids, together with an algebraic crossing condition that detects edges between the vertex sets being contracted. Third, we strengthen this to an embedded-copy theorem: a graph $M$ is a minor of $\Gamma$ if and only if, subject to analogous crossing conditions, the flow transformation monoid of $M$ is realized as the induced action of a subsemigroup of the ambient flow monoid of $\Gamma$, this subsemigroup being a monoid with a local idempotent identity.
\end{abstract}

\section{Introduction}

Flow monoids and flow transformation monoids associate to a (di)graph $\Gamma=(V,E)$ a transformation monoid generated by the elementary collapsings $\tau_{ab}$ attached to edges.
This construction was introduced by Rhodes in the course of his broad programme connecting automata, semigroups, and transformation actions generated from graph structure \cite{Rhodes2010,DomosiNehaniv2005}. Since then, this construction has been investigated from several perspectives, including structural and algorithmic properties, Green--Rees structure,  regularity, determination of maximal subgroups, Krohn-Rhodes complexity, and graph-theoretic invariants \cite{YangYang,EastGadouleauMitchell,HorvathNehanivPodoski2017}. This viewpoint also connects with a large body of work in computer science and communication networks in which graph topology constrains allowable communication, routing, or computation. Automata-network models place processors or automata at the vertices of a directed graph, with edges specifying permitted communication links; see, for example, \cite{FogelmanRobertTchuente1987,DomosiNehaniv2005}. Routing and interconnection-network questions are central in parallel and distributed computing and continue to be active in modern communication networks, including software-defined networking and graph-based methods for network control and optimization \cite{Leighton1992,DallyTowles2004,MedhiRamasamy2017,KreutzEtAl2015SDNSurvey,Jiang2022GraphDeepLearningCommunicationNetworks}.
Within this broader landscape, the Rhodes flow monoid studied here focuses on the elementary edge-collapsing transformations. This specialization is what makes graph contraction and minor structure visible in the algebra.\footnote{The digraph semigroups considered in \cite{DomosiNehaniv2005} are generated by transformations compatible with a given interconnection digraph, and the associated configuration semigroups describe induced transformations on global network states. Up to the orientation convention for edges, the flow monoid considered here is the submonoid generated only by the elementary single-vertex collapsings. Thus it is a more specialized object than the full digraph-compatible transformation semigroup, but precisely this specialization records contraction data.} 

The paper proves three related equivalences. First, connectedness of a finite non-empty graph $\Gamma$ is equivalent to the presence of constant transformations in its flow monoid $S(\Gamma)$. This algebraic test is then applied to the vertex sets contracted in a minor.  Second, the minor relation $M\leq_{\mathrm{minor}}\Gamma$ is equivalent to a division
$$
\flowTM(M)\preceq \flowTM(\Gamma)
$$
together with algebraic conditions ensuring that the state lift-sets are connected and that edge-generators of $M$ lift to products containing genuine crossing collapsings in $\Gamma$. Third, the same minor relation is equivalent to the existence of a faithful local embedded copy of $\flowTM(M)$ inside $\flowTM(\Gamma)$: one can choose representatives $R\subseteq V(\Gamma)$ and a subsemigroup $T\leq S(\Gamma)$ which is itself a monoid, with local identity an idempotent $e=e^2\in S(\Gamma)$, such that
$$
(R,T)\leq \flowTM(\Gamma)
\qquad\text{and}\qquad
(R,T)\cong \flowTM(M).
$$
The converse embedded-copy statement requires connected fibres of $e$ and crossing
elementary collapsings between the corresponding fibres. In the connected case, $e$ may
be chosen with $R=\operatorname{Im}(e)$, so that the rank of $e$ is exactly
$|V(M)|$.

The next section proves the connectedness--constant-map characterization
(Theorem~\ref{connected-constant}). Sections~\ref{sec:division} and~\ref{sec:embedding}
prove the division and embedded-copy characterizations in
Theorems~\ref{thm:minor-division} and~\ref{thm:minor-iff-embedded-copy-idem},
respectively. Theorem~\ref{thm:subsemigroup-copy} gives the forward construction used
in the embedded-copy characterization.

\section{Preliminaries and Connectivity}

A (simple) graph $\Gamma=(V,E)$ is defined to be  a vertex set $V$ and set of  undirected edges $E$, i.e., $E$ is a set of  pairs $\{x,y\}$, $x,y\in V$, $x\neq y$.  We regard undirected graphs as directed graphs for which $(x,y)\in E$ implies $(y, x)\in E$. 
In this paper, all graphs considered are finite.

\begin{definition} (\emph{Flow Monoid and Flow Transformation Monoid of a Graph}).
Given a digraph $\Gamma=(V,E)$, the \emph{flow monoid} of $\Gamma$ is the monoid $S(\Gamma)$ generated by elementary collapsings corresponding to directed edges $(a,b)$ of $\Gamma$. That is,
$$S(\Gamma)=\langle\tau_{ab}: (a,b)\in E \rangle,$$ where $\tau_{ab}: V \rightarrow V$ is a mapping called the \emph{elementary collapsing} of $(a,b)$ defined by $\tau_{ab}(a)=b$, and $\tau_{ab}(x) = x$ for all $x\neq a$.

Note that $|\Im(\tau_{ab})|=n-1$ if $|V|=n$. Also, $\tau_{ab}\tau_{ab}=\tau_{ab}$, i.e., each elementary collapsing is idempotent.
 A mapping 
$s:V\rightarrow V$ is in the flow monoid if and only if $s$ can be expressed as the product of a finite number of generators:
$$s= \tau_{x_1y_1} \dots \tau_{x_ky_k},$$ with
$k\geq 0$ and  each $(x_i, y_i)\in E$ for $1\leq i \leq k$.
Then the {\it flow transformation monoid}  $\flowTM(\Gamma)$ is $(V(\Gamma), S(\Gamma))$, with the flow monoid $S(\Gamma)$ acting on the right of the vertex set $V=V(\Gamma)$.
\end{definition}

The flow monoid always includes the identity map $1_V: V \rightarrow V$, formally identified with the product of $k=0$ generators.     By convention, we apply functions on the right of the arguments, so  $fg$ denotes the function applied by first applying $f$ and then applying $g$. Thus, we write $v \cdot fg= (v\cdot f) \cdot g, $
for all $v\in V$, $f,g\in S(\Gamma)$.  The flow monoid acts on the right of its vertex set yielding a transformation monoid.

In automata-theoretic terms, for an undirected graph $\Gamma=(V,E)$, take the input alphabet
$$
\Sigma_\Gamma=\{\tau_{ab}: a, b \in V, \{a,b\}\in E\}.
$$
One may regard $V$ as the state set of a deterministic finite automaton, with each input letter $\tau_{ab}$ acting by the corresponding elementary collapsing. The transition monoid of this automaton is precisely the flow monoid $S(\Gamma)$.
In the terminology of synchronizing automata \cite{Volkov2008},
Theorem~\ref{connected-constant} says that this automaton is synchronizing
if and only if $\Gamma$ is connected.\footnote{Indeed, if $\Gamma$ is connected and $n=|V(\Gamma)|$, the associated
automaton has reset threshold exactly $n-1$. Choose a spanning tree of $\Gamma$ rooted at $r$, and, for each non-root vertex $v$, 
let $p(v)$ be the neighbour of $v$ on the unique path in the spanning tree from $v$ to $r$. Applying the $n-1$ elementary collapsings $\tau_{v,p(v)}$, with the vertices ordered by nonincreasing distance from $r$, maps every vertex to $r$. Conversely, appending one input letter can decrease the rank of the
transformation represented by a word by at most one, so any reset word has
length at least $n-1$.}

\begin{notation}\label{not:functions-and-fibres}
For sets $X$ and $Y$, we write
$$
X^Y:=\{f:Y\to X\}
$$
for the set of all functions from $Y$ to $X$. Thus $X^X$ is the full transformation monoid on $X$.
More generally, if a semigroup or monoid $S$ acts by maps on state set $X$, the action is \emph{faithful} if distinct elements of $S$ induce distinct transformations of $X$; equivalently, for all $s,t\in S$, $$x\cdot s=x\cdot t\quad\text{for all }x\in X \quad\Rightarrow\quad
s=t.$$ 
In this case, we can identify $S$ with a subsemigroup of $X^X$, and we write $(X,S)$
for the corresponding faithful transformation semigroup or transformation monoid
(if it contains the identity map $1_X$ on $X$).

For a map $f:X\to Y$ and an element $y\in Y$, the \emph{fibre} of $f$ over $y$ is the preimage
$f^{-1}(y)$.
\end{notation}

The following elementary but key observation from Rhodes's book was formulated by C. L. Nehaniv and J. Rhodes.

\begin{lemma}\cite[Lemma~6.51b.]{Rhodes2010}
If $\tau_{ab}=\tau_{x_1x_2} \cdots \tau_{x_{2k-1}x_{2k}}$, then  $\tau_{x_1 x_2}$ is either $\tau_{ab}$ or $\tau_{ba}$.
\end{lemma}
\begin{proof}
Both sides of the equation map $a$ and $b$ to the same point. The inverse image of any other point is a singleton, since it is fixed by $\tau_{ab}$. Also, the righthand side maps $x_1$ and $x_2$ to the same point. Therefore, $\{a,b\}=\{x_1,x_2\}$.
\end{proof}

From this lemma it follows immediately that if $\tau_{ab}\in S(\Gamma)$ for some digraph $\Gamma$, then either $(a,b)\in E$ or $(b,a)\in E$. In particular if $\Gamma$ is an undirected graph, then both edges are in $\Gamma$. Then \cite[Fact~6.51a.]{Rhodes2010} follows easily from this, which we state here as a theorem for completeness. 

\begin{theorem}\cite[Fact~6.51a.]{Rhodes2010}
If $\Gamma=(V,E)$ is a graph, then $\Gamma$ is uniquely determined by the flow transformation monoid $\flowTM(\Gamma)=(V,S(\Gamma))$.
\end{theorem}
\begin{proof}
This follows easily from the fact that $(a,b), (b,a)\in E$ if and only if $\tau_{ab}\in S(\Gamma)$ for undirected graphs.
\end{proof}

\begin{theorem}[Connectedness--constant-map characterization]
\label{connected-constant}

Let $\Gamma=(V,E)$ be a non-empty finite graph, i.e., $0 < |V| < \infty$.
Then the following graph-theoretic condition is equivalent to the following algebraic conditions:

\begin{itemize}
\item[(G)] $\Gamma$ is connected.
\item[(A1)] $S(\Gamma)$ contains a constant map on $V$.
\item[(A2)] $S(\Gamma)$ contains all constant maps on $V$.
\end{itemize}

\end{theorem}

\begin{proof}

If $|V|=1$, then $S(\Gamma)=\{1_V\}$, and $1_V$ is the unique constant map on $V$.
Thus (G), (A1), and (A2) all hold.

Now suppose $|V|\geq 2$. If $\Gamma$ is connected, then since $\Gamma$ is finite there is a walk containing all of its vertices, and we can choose the walk to end with any vertex we want, say $v_0, v_1, \cdots, v_k$. Consider the product of the elementary collapsings corresponding to the edges of the walk $\tau_{v_0v_1} \tau_{v_1v_2} \cdots \tau_{v_{k-1}v_k} \in S(\Gamma)$. Denote this product by $c_{v_k}$. It is easy to verify that $c_{v_k}$ is a constant map which maps every $x\in V$ to $v_k$. Indeed, if $x=v_i$ is the first occurrence of $x$ in the walk, then the product sends
$x$ successively to $v_{i+1},v_{i+2},\ldots,v_k$.
Since we could have chosen $v_k$ to be any vertex of $\Gamma$, this shows that all the constant maps on $V$ are contained in $S(\Gamma)$.

\noindent $(A1) \rightarrow (G)$ and $(A2) \rightarrow (G)$: Take any constant map $c_v$ in $S(\Gamma)$ that maps all the vertices of $\Gamma$ to some vertex $v \in  V$. Consider a decomposition into elementary collapsings $c_v= \tau_{x_1y_1}\tau_{x_2y_2} \cdots \tau_{x_ky_k}$ with $k$ least. Since  $c_v$ must move all the vertices in $V\setminus \{v\}$ to $v$, it follows that $k\geq |V|-1 \geq 1$.  Then $y_k=v$, for if not,  then $\tau_{x_ky_k}$ is redundant in the product and we may remove it, contradicting leastness of $k$.

The elementary collapsings $\tau_{x_i y_i}$ in the product correspond to edges $(x_i,y_i)$ in $\Gamma$, and since each vertex is mapped to $v$, then it follows that there is a path from each vertex $x\in V\setminus\{v\}$ to $v$ given by following exactly those edges $(x_i,y_i)$ for the successive generators $\tau_{x_i y_i}$ in the product that actually move $x$ along its trajectory toward $v$ as we apply them from left to right. Thus the graph is connected.
\end{proof}

\begin{corollary}[Local constants on connected induced subgraphs]
\label{cor:local-constant}
Let $A\subseteq V(\Gamma)$ be nonempty and suppose that $\Gamma[A]$ is connected.
For every $r\in A$, there is an idempotent element $c_{A,r}\in S(\Gamma)$ such that
$$
a\cdot c_{A,r}=r\quad (a\in A),
\qquad
v\cdot c_{A,r}=v\quad (v\notin A).
$$
If $A=\{r\}$, then $c_{A,r}=1_{V(\Gamma)}$.
\end{corollary}

\begin{proof}
If $A=\{r\}$, take $c_{A,r}=1_{V(\Gamma)}$, the empty product of generators.  Otherwise, apply
Theorem~\ref{connected-constant} inside the connected induced subgraph $\Gamma[A]$,
and choose a product of elementary collapsings in $\Gamma[A]$ sending every vertex of
$A$ to $r$. The same product, regarded as an element of $S(\Gamma)$, fixes every vertex
outside $A$, since all its elementary collapsings have source and target in $A$. Since $c_{A,r}$ fixes $r$ and maps every vertex of $A$ to $r$, while fixing every vertex outside $A$, it is idempotent. 
\end{proof}

This theorem and its local version will be used to connect graph minors with division
of transformation monoids, since the contracted parts in a minor must be connected.
Before establishing this connection, we introduce a few more definitions.

\begin{definition}(Division of Transformation Semigroups)\label{def:division}
A transformation semigroup $(X,S)$ \emph{divides} or \emph{is emulated by} another transformation semigroup $(X',S')$, written $(X,S) \preceq (X',S')$ if there exists a subset $Y\subseteq X'$ and a subsemigroup $T$ of $S'$ such that $Y \cdot T\subseteq Y$, and there is a surjective function $\psi_1:Y \twoheadrightarrow X$ and a surjective homomorphism $\psi_2: T \twoheadrightarrow S$ such that $\psi_1(y\cdot t)=\psi_1(y)\cdot \psi_2(t)$ for every $y\in Y$ and $t\in T$.
\end{definition}

\begin{definition}[Lifts]\label{def:lifts}
Suppose $(X,S)\preceq (X',S')$ is witnessed by a subsemigroup
$T\le S'$, a surjective state map
$
\psi_1:Y\twoheadrightarrow X,
$
and a surjective homomorphism
$
\psi_2:T\twoheadrightarrow S,
$
where $Y\subseteq X'$.
For $x\in X$, the \emph{state lift-set} of $x$ is the fibre
$
\psi_1^{-1}(x)\subseteq Y
$
of the state map over $x$. Its elements are called \emph{state lifts} of $x$, and we write
$
\lifts(x):=\psi_1^{-1}(x).
$
For $s\in S$, an element $\hat{s}\in T$ is called a \emph{semigroup lift} of $s$
if
$
\psi_2(\hat{s})=s.
$
Equivalently, the semigroup lift-set of $s$ is the fibre
$
\psi_2^{-1}(s)\subseteq T.
$
\end{definition}

\begin{definition}\label{def:Gamma-v}
If $\flowTM(M) \preceq \flowTM(\Gamma)$ for graphs $M$ and $\Gamma$, then for each $v\in V(M)$ we define $\Gamma_v$ to be the subgraph of $\Gamma$ induced by the lifts of $v$, i.e.
$$\Gamma_v:=(\lifts(v), \{(v_1,v_2)\in E(\Gamma): v_1, v_2 \in \lifts(v)\}).$$
\end{definition}

The following lemma, due to Dennis Allen, Jr.\ is very helpful in proving division between transformation semigroups.

\begin{lemma}[Allen]\label{lem:Allen-division}\cite[Proposition~1.10]{DomosiNehaniv2005}
Let $(X,S)$ and $(X',S')$ be transformation semigroups. To show that $(X,S)$ divides $(X',S')$, it suffices to choose one or more $\tilde{x} \in X'$ as lifts for each $x\in X$ and one or more $\tilde{s}\in S'$ as lifts for each $s$ in a generator set for $S$, such that the following hold:
\begin{itemize}
\item[(1)] Each member of $X'$ (resp., $S'$) is a lift of at most one element of $X$ (resp., $S$).
\item[(2)] If $\tilde{x}$ is a lift of $x$ and $\tilde{s}$ is a lift of $s$, then $\tilde{x}\cdot \tilde{s}$ is some lift of $x\cdot s$.
\end{itemize}
\end{lemma}

\section{Graph Minors and Flow Transformation Monoid Division}\label{sec:division}

We use the standard minor-model characterization of the graph-minor relation, equivalent to obtaining a minor by vertex deletions, edge deletions, and edge contractions \cite{Diestel2025}.

\begin{definition}[Graph Minor]\label{def:graph-minor}
A graph $M$ is a \emph{minor} of a graph $\Gamma$ if there are pairwise disjoint sets of vertices $A_v\subseteq V(\Gamma)$, with $A_v\neq\emptyset$ for each $v\in V(M)$, such that:
\begin{itemize}
\item the subgraph of $\Gamma$ induced by $A_v$, namely $\Gamma[A_v]$, is connected; and
\item for every edge $\{u,v\}\in E(M)$, there are vertices $u'\in A_u$ and $v'\in A_v$ such that $\{u',v'\}\in E(\Gamma)$.
\end{itemize}

Such a family of sets $A_v$, together with the required edges between them, is called a \emph{realization} (or \emph{model}) of $M$ as a minor of $\Gamma$. The sets $A_v$ are called the \emph{contraction classes} of this realization.
\end{definition}

Thought of in another way, a minor $M$ is obtained from $\Gamma$ by contracting each set $A_v$, $v\in V(M)$. The contraction can be done one by one, contracting (or collapsing) one edge at a time in the connected graph $\Gamma[A_v]$ until it reduces to a single vertex. Then delete any additional vertices or edges, except the ones corresponding to edges in $M$.

\begin{remark}
Since for each $v\in V(M)$, the subgraph  $\Gamma[A_v]$ of $\Gamma$ induced by $A_v$  is non-empty and connected,  by Theorem \ref{connected-constant}, $S(\Gamma[A_v])$ contains a constant map on $A_v$.
\end{remark}

\begin{theorem}[Minor characterization by division with crossing factors]\label{thm:minor-division}
Let $M$ and $\Gamma$ be simple graphs.
Then the graph-theoretic condition that $M$ is a minor of $\Gamma$ is equivalent to the following algebraic condition:
there is a division with data
$(Y,T,\psi_1,\psi_2)$
for
$\flowTM(M)\preceq \flowTM(\Gamma)$
in the sense of Definition~\ref{def:division}, with lift sets
$$
\lifts(v):=\psi_1^{-1}(v)\qquad (v\in V(M)),
$$
such that:
\begin{itemize}
\item[(A1)] for every $v\in V(M)$, the flow monoid $S(\Gamma_v)$ contains a constant transformation,
where $\Gamma_v$ is the subgraph of $\Gamma$ induced by $\lifts(v)$; and
\item[(A2)] for each directed edge $(u,v)$ of $M$, there exists an element
$\tau^*_{uv}\in T$ with $\psi_2(\tau^*_{uv})=\tau_{uv}$ such that some factorization of
$\tau^*_{uv}$ into elementary collapsings contains a factor $\tau_{ab}$ with
$$
a\in \lifts(u),\qquad b\in \lifts(v).
$$
In particular, $\tau_{ab}$ corresponds to an edge of $\Gamma$ with one endpoint in
$\lifts(u)$ and the other in $\lifts(v)$.
\end{itemize}
\end{theorem}

\begin{proof}
\textup{(G)$\Rightarrow$(A).}
Assume that $M$ is a minor of $\Gamma$.
By Definition~\ref{def:graph-minor}, there exist pairwise disjoint vertex sets $A_v\subseteq V(\Gamma)$ for $v\in V(M)$ such that
$\Gamma[A_v]$ is connected for each $v$, and for each edge $(u,v)\in E(M)$ there exist $u'\in A_u$ and $v'\in A_v$
with $\{u',v'\}\in E(\Gamma)$.

Define
$$
Y:=\bigcup_{v\in V(M)}A_v
\qquad\text{and}\qquad
\psi_1:Y\to V(M)\ \text{by}\ \psi_1(y)=v\ \text{iff}\ y\in A_v.
$$
Then $\psi_1$ is well-defined (since the $A_v$ are disjoint) and surjective, and $\lifts(v)=\psi_1^{-1}(v)=A_v$.

Fix $v\in V(M)$. 
Choose a basepoint $v_*\in A_v \subseteq V(\Gamma)$. Since the induced graph $\Gamma[A_v]$ is connected, Theorem~\ref{connected-constant} implies 
$S(\Gamma[A_v])$ contains the constant function mapping every vertex in $A_v$ to $v_*$.   It follows that  the constant map $c_{v_*}$ on $A_v$  with value $v_*$ can  be written as a product of elementary collapsings corresponding to edges in the induced subgraph $\Gamma[A_v]$. This expression also gives a transformation in $S(\Gamma)$ that fixes all vertices outside $A_v$ while collapsing $A_v$ to $v_*$.

Also, for any $x,y\in A_v$, choose a path $x=w_1,w_2,\dots,w_r=y$ in $\Gamma[A_v]$ and define
$$
p_{xy}:=\tau_{w_1w_2}\tau_{w_2w_3}\cdots \tau_{w_{r-1}w_r}\in S(\Gamma[A_v]),
$$
so that $x\cdot p_{xy}=y$, and $p_{xy}$ fixes all vertices outside $A_v$.

Now let $(u,v)\in E(M)$ and consider the generator $\tau_{uv}\in S(M)$.
Choose $u'\in A_u$ and $v'\in A_v$ with $\{u',v'\}\in E(\Gamma)$.
Define
$$
\tau^*_{uv}:=c_{u_*}\,p_{u_*u'}\,\tau_{u'v'}\,p_{v'v_*}\,c_{v_*}\ \in S(\Gamma).
$$
Then $\tau^*_{uv}$ satisfies condition \textup{(A2)}: the displayed factorization contains the crossing factor
$\tau_{u'v'}$, with $u'\in \lifts(u)$ and $v'\in \lifts(v)$.

Let
$$T:=\langle\,\tau^*_{uv} : (u,v)\in E(M)\,\rangle_{\mathrm{mon}} \le S(\Gamma)$$
be the submonoid generated by these lifts. Since each $\tau^*_{uv}$ maps $A_u$ into $A_v$, maps $A_v$ into itself, and fixes each $A_w$ pointwise for $w\not\in\{u,v\}$, we have $Y\cdot T\subseteq Y$.

We now define $\psi_2:T\to S(M)$ and verify that it is a well-defined homomorphism.
Recall that $A_v=\psi_1^{-1}(v)$ for $v\in V(M)$, and that $Y\cdot T\subseteq Y$.

For $t\in T$ and $u\in V(M)$, choose any $y\in A_u$ and define
$$
u\cdot \psi_2(t)\ :=\ \psi_1(y\cdot t).
$$
To see that $\psi_2(t)$ is well-defined, we must show that $\psi_1(y\cdot t)$ depends only on $u=\psi_1(y)$,
not on the choice of $y\in A_u$.

It is enough to check the compatibility on the chosen lifts of the monoid generators of $S(M)$, including the identity; the general case then follows by induction on words, equivalently by Allen's Lemma~\ref{lem:Allen-division}.
Namely, for each edge $(u,v)\in E(M)$ we have chosen a lift $\tau^*_{uv}\in T$ of $\tau_{uv}\in S(M)$, and it suffices to verify:
for every $y\in Y$,
$$
\psi_1(y\cdot \tau^*_{uv})=\psi_1(y)\cdot \tau_{uv}.
$$
Let $y\in Y$. If $y\in A_u$, then $y\cdot\tau^*_{uv}\in A_v$, hence
$$
\psi_1(y\cdot\tau^*_{uv})=v=\psi_1(y)\cdot\tau_{uv}.
$$
If $y\in A_w$ with $w\neq u$, then $y\cdot\tau^*_{uv}\in A_w$: for $w\not\in\{u,v\}$ the factors of $\tau^*_{uv}$ all fix 
$A_w$ pointwise, while for $w=v$ the factors supported in $A_v$ may move $y$ but keep it inside $A_v$. Hence
$$
\psi_1(y\cdot\tau^*_{uv})=w=\psi_1(y)\cdot\tau_{uv}.
$$
Since $T$ is a submonoid, it contains $1_{V(\Gamma)}$, which serves as the lift of the identity of $S(M)$; clearly
$$
\psi_1(y\cdot 1_{V(\Gamma)})=\psi_1(y)=\psi_1(y)\cdot 1_{V(M)} .
$$

Therefore the compatibility condition of Lemma~\ref{lem:Allen-division}(2) holds for the chosen lifts of vertices and monoid generators, including the identity.
Consequently, $\psi_2$ is a well-defined surjective homomorphism. Indeed, $\psi_2(\tau^*_{uv})=\tau_{uv}$ for all $(u,v)\in E(M)$,
and the elements $\tau_{uv}$, together with the identity, generate the monoid $S(M)$; the identity $1_{V(M)}$  is lifted to $1_{V(\Gamma)}\in T$.
Moreover,
$$
\psi_1(y\cdot t)=\psi_1(y)\cdot\psi_2(t)\qquad (y\in Y,\ t\in T),
$$
so indeed $\flowTM(M)\preceq \flowTM(\Gamma)$ (Definition~\ref{def:division}).\\

\noindent
\textup{(A)$\Rightarrow$(G).}
Assume there is a semigroup division with $Y \subseteq V(\Gamma)$ and $T \leq S(\Gamma)$ with surjective function $\psi_1: Y\twoheadrightarrow V(M)$ and 
 surjective homomorphism $\psi_2: T\twoheadrightarrow S(M)$  satisfying \textup{(A1)}--\textup{(A2)}.
Thus  $Y\cdot T\subseteq Y$, and 
$$
\psi_1(y\cdot t)=\psi_1(y)\cdot \psi_2(t)\qquad (y\in Y,\ t\in T).
$$
For each $v\in V(M)$, set $A_v:=\lifts(v)=\psi_1^{-1}(v)\subseteq Y$. Then the $A_v$ are pairwise disjoint.

By \textup{(A1)}, $S(\Gamma_v)$ contains a constant transformation for each $v$.
By Theorem~\ref{connected-constant}, $\Gamma_v$ is connected; by Definition~\ref{def:Gamma-v}, this means precisely that
the induced subgraph $\Gamma[A_v]$ is connected. Hence the first bullet of Definition~\ref{def:graph-minor} holds for the family $\{A_v\}_{v\in V(M)}$.

Now let $(u,v)\in E(M)$. By \textup{(A2)}, there is a lift $\tau^*_{uv}$ of $\tau_{uv}$ whose factorization contains a crossing factor $\tau_{ab}$
with $a\in A_u$ and $b\in A_v$. Since $\tau_{ab}$ is an elementary collapsing in $S(\Gamma)$, we have $\{a,b\}\in E(\Gamma)$.
Thus there is an edge of $\Gamma$ with one endpoint in $A_u$ and the other in $A_v$, establishing the second bullet of
Definition~\ref{def:graph-minor}.

Therefore $\{A_v:v\in V(M)\}$ satisfies both conditions of Definition~\ref{def:graph-minor}, so $M$ is a minor of $\Gamma$.
\end{proof}

\section{Graph Minors and Embedded Flow Transformation Monoids}\label{sec:embedding}

Next in Theorem~\ref{thm:subsemigroup-copy}, we strengthen the division result of Theorem~\ref{thm:minor-division} by constructing an embedded copy $(R,T)$ of the flow transformation monoid  $\flowTM(M)$ of the minor $M$ inside the flow transformation monoid $\flowTM(\Gamma)$ of the graph $\Gamma$ satisfying certain algebraic conditions. These will guarantee a converse result that allows for  the detection of embedded transformation monoids corresponding precisely to minors of $\Gamma$ (Theorem~\ref{thm:minor-iff-embedded-copy-idem}).

\begin{theorem}[Embedded flow transformation monoid of a minor]
\label{thm:subsemigroup-copy}
Let $M$ and $\Gamma$ be simple graphs, and suppose that $M$ is realized
as a minor of $\Gamma$ via pairwise disjoint connected contraction classes
$A_x\subseteq V(\Gamma)$ for $x\in V(M)$. Choose representatives $r_x\in A_x$, and set
$$
R:=\{r_x:x\in V(M)\}.
$$
Then there exist an idempotent $e\in S(\Gamma)$ and a subsemigroup
$T\le S(\Gamma)$ such that $T$ is itself a monoid with identity $e$,
$R\cdot T\subseteq R$, and an isomorphism of transformation monoids
$$
(R,T)\cong \flowTM(M) .
$$
Moreover, for each directed edge $(x,y)$ of $M$, the element of $T$ corresponding
to $\tau_{xy}\in S(M)$ has a factorization containing an elementary collapsing
$$
\tau_{ab}
\quad\text{with}\quad
a\in A_x,\quad b\in A_y.
$$
\end{theorem}

\begin{proof}
\emph{Step 0: Local routing maps and the class-collapse idempotent.}
For each $x\in V(M)$, Corollary~\ref{cor:local-constant} gives a local
constant
$$
c_x:=c_{A_x,r_x}\in S(\Gamma)
$$
which sends every vertex of $A_x$ to $r_x$ and fixes every vertex outside $A_x$.
If $A_x=\{r_x\}$, then $c_x=1_{V(\Gamma)}$. Since the sets $A_x$ are pairwise
disjoint, the elements $c_x$ commute. Define
$$
e:=\prod_{x\in V(M)}c_x .
$$
Then $e$ is a class-collapse idempotent that maps each contraction class $A_x$ to $r_x$, and fixes
$$
C:=V(\Gamma)\setminus\bigcup_{x\in V(M)}A_x
$$
pointwise.

Since $\Gamma[A_x]$ is connected, for $u,v\in A_x$ there is a path $w_0 = u, w_1, \cdots, w_k = v$ in $A_x$. The product of the elementary collapsings along this path 
$$
p^{(x)}_{u\to v} = \tau_{w_0w_1} \tau_{w_1w_2} \cdots \tau_{w_{k-1}w_k} \in S(\Gamma)
$$
acts as $w \cdot p^{(x)}_{u\to v} = v$ for $w \in \{w_0, \dots , w_k\}$ and $w \cdot p^{(x)}_{u\to v} = w$ otherwise. We call such a product a local routing word from $u$ to $v$ inside $A_x$. 

\emph{Step 1: Lift the generators of $S(M)$.}
For each directed edge $(x,y)$ of $M$, choose a crossing edge
$$
\{a_{xy},b_{xy}\}\in E(\Gamma),
\qquad
a_{xy}\in A_x,\quad b_{xy}\in A_y,
$$
and define
$$
\gamma_{xy}:=e\,p^{(x)}_{r_x\to a_{xy}}\,\tau_{a_{xy}b_{xy}}\,e\in S(\Gamma).
$$
Since $\gamma_{xy}\in eS(\Gamma)e$, we have
$$
e\gamma_{xy}=\gamma_{xy}e=\gamma_{xy}.
$$

We compute the action of $\gamma_{xy}$ on $R$. Starting at $r_x$, the leftmost
$e$ fixes $r_x$, the routing word sends $r_x$ to $a_{xy}$, the crossing
collapsing sends $a_{xy}$ to $b_{xy}\in A_y$, and the final $e$ sends $b_{xy}$
to $r_y$. Thus
$$
r_x\cdot \gamma_{xy}=r_y.
$$
If $z\neq x$, then $r_z\notin A_x$, so the routing word fixes $r_z$, the crossing
collapsing fixes $r_z$, and both copies of $e$ fix $r_z$. Hence
$$
r_z\cdot \gamma_{xy}=r_z\qquad (z\neq x).
$$
Therefore, on $R$, $\gamma_{xy}$ acts as the elementary collapsing
$$
\tau_{r_xr_y}.
$$

\emph{Step 2: Generate the embedded local monoid.}
Let
$$
T:=\langle e,\gamma_{xy}:(x,y)\in E(M)\rangle\le S(\Gamma)
$$
be the subsemigroup generated by $e$ and these lifted edge-collapsings. Since $e$ is
a two-sided identity for each generator $\gamma_{xy}$, it is the identity element of $T$.
Thus $T$ is itself a monoid, with identity $e$. Also, each generator maps $R$ into
$R$, so $R\cdot T\subseteq R$.

{\it Step 3:} We now check that the action of $T$ on $R$ is exactly the flow monoid action of $S(M)$
on $V(M)$, under the bijection $x\mapsto r_x$, and that this action is faithful. The element $e$ acts as the identity on $R$, while each $\gamma_{xy}$
acts as $\tau_{r_xr_y}$. Hence, under the bijection
$$
V(M)\rightarrow R,\qquad x\mapsto r_x,
$$
the induced action of $T$ on $R$ is generated by the identity on $R$ and the elementary
collapsings $\tau_{r_xr_y}$ corresponding to the directed edges $(x,y)$ of $M$. Hence
this induced action is precisely the flow monoid action of $S(M)$, transported to $R$.

{\it Step 4:} It remains only to prove faithfulness of the action of $T$ on $R$. Every generator $\gamma_{xy}$ fixes $C$
pointwise: the idempotent $e$ fixes $C$, the routing word is supported in $A_x$,
and the crossing factor has both endpoints in $A_x\cup A_y$. The idempotent $e$
also fixes $C$ pointwise. Hence every $t\in T$ fixes $C$ pointwise.

Suppose $t_1,t_2\in T$ agree on $R$. Since $e$ is the identity of $T$, we have
$et_i=t_i$ for $i=1,2$.\\
 If $v\in A_x$ for some $x\in V(M)$, then $v\cdot e=r_x$, and so
$$
v\cdot t_i
=
v\cdot et_i
=
(v\cdot e)\cdot t_i
=
r_x\cdot t_i .
$$
Since $t_1$ and $t_2$ agree on $R$, this gives
$$
v\cdot t_1=r_x\cdot t_1=r_x\cdot t_2=v\cdot t_2.
$$
Otherwise, $v\in C$, and both $t_1$ and $t_2$ fix $v$. Thus $t_1$ and $t_2$ agree
on all of $V(\Gamma)$, i.e., $t_1 = t_2$. Hence the action of $T$ on $R$ is faithful.

Therefore, we have an isomorphism of transformation monoids
$$
(R,T)\cong \flowTM(M).
$$
The displayed definition of $\gamma_{xy}$ gives the required factorization containing
the crossing factor $\tau_{a_{xy}b_{xy}}$.
\end{proof}

\begin{remark}\label{rem:constructed-copy}
The construction in Theorem~\ref{thm:subsemigroup-copy} depends on the chosen
realization of $M$ as a minor of $\Gamma$, on the representatives, and on the chosen
crossing edges. Its conclusion is stronger than division: it gives an embedded copy
$$
(R,T)\cong \flowTM(M)
$$
inside the ambient flow transformation monoid $\flowTM(\Gamma)$. The acting
semigroup $T\le S(\Gamma)$ is a subsemigroup which is itself a monoid, with local
identity $e$, not necessarily the ambient identity $1_{V(\Gamma)}$.
\end{remark}

For $F\subseteq V(\Gamma)$, let $S_\Gamma(F)$ denote the monoid generated by the
restrictions to $F$ of those elementary collapsings $\tau_{ab}$ with $a,b\in F$ that
occur among the generators of $S(\Gamma)$. Equivalently, $S_\Gamma(F)$ is the flow
monoid of the induced graph $\Gamma[F]$, defined via the elementary collapsings inherited
from $S(\Gamma)$.

\begin{theorem}[Embedded-copy characterization of minors]
\label{thm:minor-iff-embedded-copy-idem}
Let $M$ and $\Gamma$ be simple graphs. Then the following geometric and algebraic
conditions are equivalent:

\begin{itemize}
\item[\textup{(G)}] $M$ is a minor of $\Gamma$.
\item[\textup{(A)}] There exist an idempotent $e\in S(\Gamma)$, a subsemigroup
$T\leq S(\Gamma)$ with identity element $e$, and a subset
$R\subseteq\operatorname{Im}(e)$,
such that the following hold:
\begin{itemize}
\item[\textup{(i)}] $(R,T)\cong \flowTM(M)$ via a state bijection
$\psi_1:R\rightarrow V(M)$ and monoid isomorphism $\psi_2:T\rightarrow S(M)$;

\item[\textup{(ii)}] for each $x\in V(M)$, writing
$$
r_x:=\psi_1^{-1}(x)
\qquad\text{and}\qquad
F_x:=e^{-1}(r_x),
$$
the monoid $S_\Gamma(F_x)$ contains a constant transformation;

\item[\textup{(iii)}] for every directed edge $(x,y)$ of $M$, the element
$$
\gamma_{xy}:=\psi_2^{-1}(\tau_{xy})\in T
$$
has a factorization into elementary collapsings from the generating set of
$S(\Gamma)$ containing a factor $\tau_{ab}$ with
$$
a\in F_x,
\qquad
b\in F_y.
$$
\end{itemize}
\end{itemize}
\end{theorem}

\begin{proof}
\textup{(G)$\Rightarrow$(A).}
Assume that $M$ is a minor of $\Gamma$. The proof in this direction is essentially the application
Theorem~\ref{thm:subsemigroup-copy} for a realization of $M$ as a minor of
$\Gamma$. Let $e$, $R$, and $T$ be the idempotent, representative set, and embedded
local monoid constructed there. 

Let
$$
\psi_1:R\rightarrow V(M)
$$
be the state bijection defined by $\psi_1(r_x)=x$. Theorem~\ref{thm:subsemigroup-copy}
gives an isomorphism of transformation monoids
$
(R,T)\cong \flowTM(M).
$
Let
$
\psi_2:T\rightarrow S(M)
$
be the corresponding monoid isomorphism. Thus condition \textup{(i)} holds.

For each $x\in V(M)$, we have $r_x=\psi_1^{-1}(x)$, and the fibre
$F_x=e^{-1}(r_x)$ 
is exactly the contraction class $A_x$ in the chosen realization. Hence
$\Gamma[F_x]$ is connected, and therefore $S_\Gamma(F_x)$ contains a constant
transformation by Theorem~\ref{connected-constant}. Thus condition \textup{(ii)}
holds.

Finally, let $(x,y)$ be a directed edge of $M$. In the construction of
Theorem~\ref{thm:subsemigroup-copy}, the element of $T$ corresponding to
$\tau_{xy}\in S(M)$ is
$\gamma_{xy}=\psi_2^{-1}(\tau_{xy}).$
That theorem gives a factorization of $\gamma_{xy}$ into elementary collapsings from
the generating set of $S(\Gamma)$ containing a crossing factor
$\tau_{ab}$
with $a\in A_x=F_x, \ b\in A_y=F_y$. Thus condition \textup{(iii)} holds.

\smallskip
\noindent
\textup{(A)$\Rightarrow$(G).}
Assume \textup{(A)}. For each $x\in V(M)$, write
$$
r_x:=\psi_1^{-1}(x)
\qquad\text{and}\qquad
F_x:=e^{-1}(r_x).
$$
The sets $F_x$ are pairwise disjoint since they are fibres of the map $e$. They are
nonempty since $r_x\in R\subseteq\operatorname{Im}(e)$ and $e$ is idempotent:
indeed, $r_x\cdot e=r_x$, so $r_x\in F_x$. By condition \textup{(ii)} and
Theorem~\ref{connected-constant}, each induced subgraph $\Gamma[F_x]$ is connected.

Now let $(x,y)$ be a directed edge of $M$. By condition \textup{(iii)}, the element
$$
\gamma_{xy}:=\psi_2^{-1}(\tau_{xy})\in T
$$
has a factorization into elementary collapsings from the generating set of $S(\Gamma)$
containing a factor
$$
\tau_{ab}
\quad\text{with}\quad
a\in F_x,\qquad b\in F_y.
$$
Since this factor $\tau_{ab}$ is one of the generating elementary collapsings of
$S(\Gamma)$, we have $\{a,b\}\in E(\Gamma)$. Hence there is an edge of $\Gamma$
between $F_x$ and $F_y$.

Therefore the pairwise disjoint connected sets $F_x\subseteq V(\Gamma)$, for
$x\in V(M)$, are contraction classes realizing $M$ as a minor of $\Gamma$.

Equivalently, the same hypotheses determine division data of the kind used in
Theorem~\ref{thm:minor-division}; the argument above spells out directly the resulting
contraction classes.
\end{proof}

\begin{corollary}[Choosing the idempotent with image size $|V(M)|$ in the connected case]
\label{cor:connected-image-R}
Assume $\Gamma$ is connected and $M$ is a minor of $\Gamma$. Then the idempotent
$e$ in Theorem~\ref{thm:minor-iff-embedded-copy-idem} may be chosen so that
$\operatorname{Im}(e)=R.$
Equivalently, the fibres $e^{-1}(r)$, $r\in R$, form a partition of $V(\Gamma)$. In this case, $|\operatorname{Im}(e)|=|V(M)|$.
\end{corollary}

\begin{proof}
Fix a realization of $M$ as a minor of $\Gamma$, given by pairwise disjoint nonempty
connected sets $A_x\subseteq V(\Gamma)$ for $x\in V(M)$. Let
$$
A:=\bigcup_{x\in V(M)}A_x .
$$
Since $\Gamma$ is connected, every vertex of $V(\Gamma)\setminus A$ is connected
to $A$ by a path. Choose a forest $F$ in $\Gamma$, rooted at the vertices of $A$,
spanning $V(\Gamma)$, such that every vertex of $V(\Gamma)\setminus A$ has a unique
path in $F$ to a vertex of $A$. For each $x\in V(M)$, define
$$
A'_x
:=
A_x\cup
\{v\in V(\Gamma)\setminus A:\text{ the path in $F$ from $v$ to $A$ ends in }A_x\}.
$$
Then the sets $A'_x$ are pairwise disjoint, cover $V(\Gamma)$, and each induced
subgraph $\Gamma[A'_x]$ is connected. The original crossing edges between the $A_x$
remain crossing edges between the enlarged sets $A'_x$.

Retain the chosen representatives $r_x\in A_x\subseteq A'_x$. For each $x$, choose
a local constant idempotent $e_x\in S(\Gamma)$ which maps every vertex of $A'_x$
to $r_x$ and fixes every vertex outside $A'_x$. Since the sets $A'_x$ are disjoint,
the idempotents $e_x$ commute, and their product
$$
e:=\prod_{x\in V(M)}e_x
$$
is an idempotent. It collapses each $A'_x$ to $r_x$, and since the sets $A'_x$
cover $V(\Gamma)$, it has image
$$
\operatorname{Im}(e)=\{r_x:x\in V(M)\}=R.
$$
Applying Theorem~\ref{thm:subsemigroup-copy} to this enlarged realization gives the
required embedded local copy with $\operatorname{Im}(e)=R$, which has $|V(M)|$ elements.
\end{proof}
\subsection*{Acknowledgments}
\phantomsection
\addcontentsline{toc}{subsection}{Acknowledgments}
\label{Acknowledgements}

We gratefully acknowledge the support of the Natural Sciences and Engineering Research Council of Canada (NSERC), funding reference number RGPIN-2019-04669. 
Cette recherche a \'et\'e financ\'ee par le Conseil de recherches en sciences naturelles et en g\'enie du Canada (CRSNG), num\'ero de r\'ef\'erence RGPIN-2019-04669.  We thank alumni of the Waterloo Algebraic Intelligence \& Computation Lab, especially Marta Zheplinska and Meenal Gupta, for their involvement in our early discussions of algebraic properties of  graphs and their minors.

\bibliographystyle{eptcs}
\bibliography{ref}

\end{document}